%% file: main.tex
\documentclass[a4paper,USenglish,cleveref, autoref, thm-restate, authorcolumns]{lipics-v2021}

\hideLIPIcs  %

\newcommand{\sv}[1]{}%
 \newcommand{\lv}[1]{#1}%
\usepackage{etoolbox}
\newcommand{\appendixText}{}

 \newcommand{\toappendix}[1]{#1}%

\input{header}
\title{Max Weight Independent Set in graphs with no long claws: An analog of the Gy\'{a}rf\'{a}s' path argument}
\titlerunning{Max Independent Set in graphs with no long claws}

\lv{
\author{Konrad Majewski}{Institute of Informatics, Faculty of Mathematics, Informatics and Mechanics, University of Warsaw}{}{}{}
\author{Tom\'a\v{s} Masa\v{r}\'ik}{Institute of Informatics, Faculty of Mathematics, Informatics and Mechanics, University of Warsaw}{masarik@mimuw.edu.pl}{0000-0001-8524-4036}{}
\author{Jana Novotn\'a}{Institute of Informatics, Faculty of Mathematics, Informatics and Mechanics, University of Warsaw}{jnovotna@mimuw.edu.pl}{0000-0002-7955-4692}{}
\author{Karolina Okrasa}{ Faculty of Mathematics and Information Science, Warsaw University of Technology \and Institute of Informatics, Faculty of Mathematics, Informatics and Mechanics, University of Warsaw}{k.okrasa@mini.pw.edu.pl}{0000-0003-1414-3507}{}
\author{Marcin Pilipczuk}{Institute of Informatics, Faculty of Mathematics, Informatics and Mechanics, University of Warsaw}{m.pilipczuk@mimuw.edu.pl}{0000-0001-5680-7397}{}
\author{Paweł Rzążewski}{Faculty of Mathematics and Information Science, Warsaw University of Technology \and Institute of Informatics, Faculty of Mathematics, Informatics and Mechanics, University of Warsaw}{p.rzazewski@mini.pw.edu.pl}{0000-0001-7696-3848}{Partially supported by Polish National Science Centre grant no. 2018/31/D/ST6/00062.}
\author{Marek Sokołowski}{Institute of Informatics, Faculty of Mathematics, Informatics and Mechanics, University of Warsaw}{marek.sokolowski@mimuw.edu.pl}{}{}

\authorrunning{Majewski, Masa\v{r}\'ik, Novotn\'{a}, Okrasa,  Pilipczuk,  Rzążewski, and  Sokołowski} 

\Copyright{Konrad Majewski, Tom\'{a}\v{s} Masa\v{r}\'ik, Jana Novotn\'{a}, Karolina Okrasa, Marcin Pilipczuk, Paweł Rzążewski, and Marek Sokołowski}
}

\sv{
\author{anonym}{anonym}{}{}{}
}



\category{} %

\relatedversion{} %

\funding{\flag{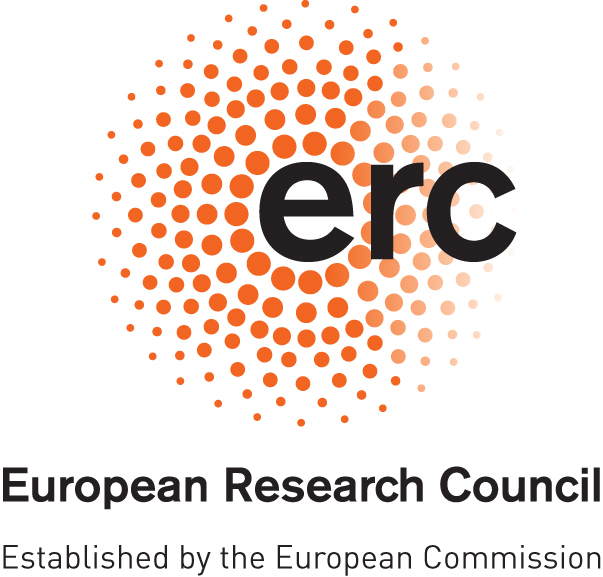}\flag{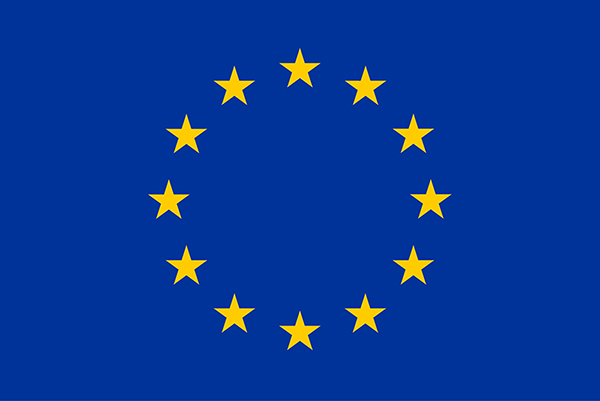}This research is part of projects that has received funding from the European Research Council (ERC)
under the European Union's Horizon 2020 research and innovation programme
Grant Agreement 714704 (JN, KO, MP, PRz) and 948057 (KM, TM, MS).}

\nolinenumbers %

\EventEditors{John Q. Open and Joan R. Access}
\EventNoEds{2}
\EventLongTitle{42nd Conference on Very Important Topics (CVIT 2016)}
\EventShortTitle{CVIT 2016}
\EventAcronym{CVIT}
\EventYear{2016}
\EventDate{December 24--27, 2016}
\EventLocation{Little Whinging, United Kingdom}
\EventLogo{}
\SeriesVolume{42}
\ArticleNo{23}

\usepackage[normalem]{ulem}

\usepackage{mathtools}

\begin{document}

\maketitle

\begin{abstract}
\input{abstract}

\end{abstract}

\lv{ \newpage}
\section{Introduction}\label{sec:intro}
\input{intro}

\section{Preliminaries}\label{sec:prelim}
\input{prelim}

\section{Main result}\label{sec:main}
\input{main-res}

\section{Algorithmic applications}\label{sec:algo}
\input{algorithmic}

\section{Conclusion}\label{sec:conclusion}
\input{conc}

\bibliography{main}

\sv{
  \newpage
\appendix
\section{Appendix}
\appendixText
}

\end{document}

%% file: header.tex
\usepackage{todonotes}
\usepackage{xspace}

\newcommand{\Oh}{\mathcal{O}}

\newcommand{\eps}{\epsilon}

\newcommand{\NP}{\textsf{NP}\xspace}
\newcommand{\APX}{\textsf{APX}\xspace}
\newcommand{\Z}{\mathbb{Z}_{\geq 0}}
\newcommand{\cA}{\mathcal{A}}
\newcommand{\cT}{\mathcal{T}}

\renewcommand{\leq}{\leqslant}
\renewcommand{\geq}{\geqslant}
\renewcommand{\epsilon}{\varepsilon}
\renewcommand{\setminus}{-}

\newcommand{\wei}{\mathfrak{w}}

\newcommand{\exponentqptaslogn}{5}
\newcommand{\exponentqptaseps}{-1}

%% file: abstract.tex
We revisit recent developments for the \textsc{Maximum Weight Independent Set} problem
in graphs excluding a subdivided claw $S_{t,t,t}$ as an induced subgraph
[Chudnovsky, Pilipczuk, Pilipczuk, Thomass\'{e}, SODA 2020] 
and provide a subexponential-time algorithm with improved running time $2^{\Oh(\sqrt{n}\log n)}$
and a quasipolynomial-time approximation scheme with improved
running time $2^{\Oh(\varepsilon^{\exponentqptaseps} \log^{\exponentqptaslogn} n)}$.

The Gy\'arf\'as' path argument,
a powerful tool that is the main building block for many algorithms in $P_t$-free graphs,
ensures that given an $n$-vertex $P_t$-free graph, in polynomial time we can find a set $P$ of at most $t-1$ vertices,
such that every connected component of $G-N[P]$ has at most $n/2$ vertices.
Our main technical contribution is an analog of this result for $S_{t,t,t}$-free graphs:
given an $n$-vertex $S_{t,t,t}$-free graph, in polynomial time we can find a set $P$ of $\Oh(t \log n)$ vertices
and an extended strip decomposition (an appropriate analog of the decomposition into connected components)
of $G-N[P]$ such that every particle (an appropriate analog of a connected component to recurse on) 
of the said extended strip decomposition has at most $n/2$ vertices.

%% file: intro.tex
The complexity of the \textsc{Maximum Weight Independent Set} problem (\textsc{MWIS} for short), one of the classic combinatorial optimization problems,
varies depending on the restrictions imposed on the input graph from polynomial-time solvable (e.g., in bipartite or chordal graphs)
through known to admit a quasipolynomial-time algorithm (graphs with bounded longest induced path~\cite{GartlandL20}),
a polynomial-time approximation scheme and a fixed-parameter algorithm (planar graphs~\cite{Baker94}), 
a quasipolynomial-time approximation scheme (graphs excluding a fixed subdivided claw as an induced subgraph~\cite{DBLP:journals/corr/abs-1907-04585, DBLP:conf/soda/ChudnovskyPPT20}), 
to being \NP-hard and hard to approximate within $n^{1-\varepsilon}$ factor in general graphs~\cite{Hastad99,Zuckerman07}.
A methodological study of this behavior leads to the following question:
\begin{quote}
For which structures in the input graph, the assumption of their absence from the input
graph makes \textsc{MWIS} easier and by how much?
\end{quote}
The ``absence of structures'' notion can be made precise by specifying the forbidden structure
and the containment relation, for example as a minor, topological minor, induced minor, subgraph,
or induced subgraph. The last one --- induced subgraph relation --- is the weakest one, and thus
the most expressible. %
This leads to the study of the complexity of \textsc{MWIS} in various hereditary graph classes,
that is, graph classes closed under vertex deletion and thus definable
by a (possibly infinite) list of forbidden induced subgraphs. 

While a general classification of \emph{all} hereditary graph classes with regards to the complexity of \textsc{MWIS} (or other classic graph problems) may be too complex, 
classifying graph classes with one forbidden induced subgraph looks more feasible.
That is, we focus on $H$-free graphs, graphs excluding a fixed graph $H$ as an induced subgraph.
Furthermore, the complexity of a given problem (here, \textsc{MWIS}) in $H$-free graphs
may indicate the impact of forbidding $H$ as an induced subgraph on the complexity
of \textsc{MWIS} in more general settings.

As observed by Alekseev~\cite{alekseev1982effect,Alekseev03},
the fact that \textsc{MWIS} remains \NP-hard and \APX-hard in subcubic graphs,
together with the observation that subdividing every edge twice in a graph increases
the size of the maximum independent set by exactly the number of edges of the original graph,
leads to the conclusion that \textsc{MWIS} remains \NP-hard and \APX-hard in $H$-free graphs
unless every connected component of $H$ is a path or a tree with three leaves. 

In what follows, for integers $t,a,b,c > 0$, by $P_t$ we denote the path on $t$
vertices, and by $S_{a,b,c}$ we denote the tree with three leaves within distance $a$, $b$, and $c$ from the unique vertex of degree $3$ of the tree. 
Since 1980s, it has been known that \textsc{MWIS} is polynomial-time solvable in $P_4$-free graphs (because of their strong structural properties)
and in $S_{1,1,1}$-free graphs~\cite{MINTY1980284,SBIHI198053} (because the notion of an augmenting path from the matching problem generalizes to \textsc{MWIS} in $S_{1,1,1}$-free, i.e., claw-free graphs). 
For many years, only partial results in subclasses were obtained until the area started to develop rapidly around 2014. 

Lokshtanov, Vatshelle, and Villanger~\cite{LokshtanovVV14} adapted the framework of
potential maximal cliques~\cite{BouchitteT01} to show a polynomial-time algorithm
for \textsc{MWIS} in $P_5$-free graphs; this was later generalized to $P_6$-free graphs~\cite{GrzesikKPP19} and other related graph classes~\cite{AbrishamiCDTTV22,AbrishamiCPRS21}.
More importantly for this work, Bacs\'o et al.~\cite{DBLP:journals/algorithmica/BacsoLMPTL19} observed that the classic Gy\'{a}rf\'{a}s' path argument,
developed to show that for every fixed $t$ the class of $P_t$-free graphs is $\chi$-bounded~\cite{gyarfas1,gyarfas2}, 
also easily gives a subexponential-time algorithm for \textsc{MWIS} in $P_t$-free graphs. 
The crucial corollary of the Gy\'{a}rf\'{a}s' path argument lies in the following.
\begin{theorem}\label{thm:gyarfas}
Given an $n$-vertex graph $G$, one can in polynomial time find an induced path $Q$ in $G$
such that every connected component of $G-N[V(Q)]$ has at most $n/2$ vertices.
\end{theorem}
For $P_t$-free graphs the said path $Q$ has at most $t-1$ vertices. 
Bacs\'o et al.~\cite{DBLP:journals/algorithmica/BacsoLMPTL19}
observed that branching either on the highest degree vertex
(if this degree is larger than $\sqrt{n}$) or on the whole set $N[V(Q)]$ for the path $Q$ coming from \cref{thm:gyarfas}
(otherwise)
gives an algorithm with running time bound exponential in $\sqrt{n} \cdot \mathrm{poly}(t,\log n)$. 

Chudnovsky, Pilipczuk, Pilipczuk, and Thomass\'{e}~\cite{DBLP:journals/corr/abs-1907-04585, DBLP:conf/soda/ChudnovskyPPT20} added to the mix 
an observation that a simple branching algorithm is able to get rid of \emph{heavy} vertices: vertices of the input graph whose neighborhood contains a large
fraction of the sought independent set. 
Once this branching is executed and the graph does not have heavy vertices, 
the set $N[Q]$ from \cref{thm:gyarfas} contains only a small fraction of the sought solution and, if one aims for an approximation algorithm, can be just sacrificed, yielding 
a quasipolynomial-time approximation scheme (QPTAS) for \textsc{MWIS} in $P_t$-free graphs.
Using this as a starting point and leveraging on the celebrated \emph{three-in-a-tree} theorem of Chudnovsky and Seymour~\cite{DBLP:journals/combinatorica/ChudnovskyS10}, they developed a much more involved QPTAS and a subexponential
algorithm (with running time bound $2^{n^{8/9} \mathrm{poly}(\log n, t)}$) for \textsc{MWIS} in $S_{t,t,t}$-free graphs. 

Consider the following simple template for a branching algorithm for \textsc{MWIS}: 
if the current graph is disconnected, solve independently every connected component; otherwise, 
pick a vertex (\emph{pivot}) $v$ and branch whether $v$ is in the sought independent set (recursing on $G-N[v]$)
or not (recursing on $G-v$). 
The performance of such an algorithm highly depends on how we choose the pivot $v$. 
\cref{thm:gyarfas} suggests that in $P_t$-free graphs the vertices of $Q$ may be good choices:
there is only a bounded number of them, and the deletion of \emph{the whole} neighborhood $N[V(Q)]$ splits $G$
into multiplicatively smaller pieces. 
In a breakthrough result, Gartland and Lokshtanov~\cite{GartlandL20} 
showed how to choose the pivot and measure the progress of the algorithm, obtaining a quasipolynomial-time
algorithm for \textsc{MWIS} in $P_t$-free graphs. 
Later, Pilipczuk, Pilipczuk, and Rz\k{a}\.{z}ewski~\cite{PilipczukPR21} provided an arguably simpler measure, leading
to an improved (but still quasipolynomial) running time bound. 
These developments have been subsequently generalized to a larger class of problems beyond \textsc{MWIS} and to
$C_{>t}$-free graphs (graphs without induced cycle of length more than $t$)~\cite{GartlandLPPR21}.

This progress suggests that \textsc{MWIS} may be actually solvable in polynomial time in $H$-free graphs
for all open cases, that is, whenever $H$ is a forest whose every connected component has at most three leaves. 
However, we seem still far from proving it: not only we do not know how to improve the quasipolynomial bounds
of~\cite{GartlandL20,PilipczukPR21} to polynomial ones, but also it remains unclear how to merge
the approach of~\cite{GartlandL20,PilipczukPR21} with the way how~\cite{DBLP:journals/corr/abs-1907-04585, DBLP:conf/soda/ChudnovskyPPT20}
used the three-in-a-tree theorem~\cite{DBLP:journals/combinatorica/ChudnovskyS10}. 

\medskip

In this work, we make a step in this direction, providing an analog of \cref{thm:gyarfas}
for $S_{t,t,t}$-free graphs. 
Before we state it, let us briefly discuss what we can hope for in the class of $S_{t,t,t}$-free graphs.

Consider an example of a graph $G$ being the line graph of a clique $K$. 
The graph $G$ is $S_{1,1,1}$-free, but does not admit any (balanced in any useful sense) separator of the form $N[P]$ for a small set $P \subseteq V(G)$. 
The \textsc{MWIS} problem on $G$ translates back to the maximum weight matching problem in the clique $K$;
this problem is polynomial-time solvable, but with very different methods than branching. In particular, 
we are not aware of any way of solving maximum weight matching in a clique in quasipolynomial time by simple branching. 
Thus, we expect that an algorithm for \textsc{MWIS} in $S_{t,t,t}$-free graphs, given such a graph $G$,
will discover that it is actually working with the line graph of a clique and apply maximum weight matching techniques
to the preimage graph $K$.

Chudnovsky and Seymour, in their project to understand claw-free graphs~\cite{claw-free-survey}, developed a good way of describing 
that a graph ``looks like a line graph'' by the notion of an \emph{extended strip decomposition}.
The formal definition can be found in \cref{sec:prelim}. Here, we remark that
in an extended strip decomposition of a graph, one can distinguish \emph{particles} being induced
subgraphs of the graph; an algorithm for \textsc{MWIS}
can recurse on individual particles, compute the maximum weight independent sets there, and combine the results
into a maximum weight independent set in the whole graph using a maximum weight matching algorithm
on an~auxiliary graph (cf.~\cite{DBLP:journals/corr/abs-1907-04585, DBLP:conf/soda/ChudnovskyPPT20}). 
Thus, an extended strip decomposition of a graph with particles of multiplicatively smaller size 
is very useful for recursion; it can be seen as an analog of splitting into connected components of
multiplicatively smaller size, as it is in the case of the components of $G-N[V(Q)]$ in \cref{thm:gyarfas}.

With the above discussion in mind, we can now state our main technical result.
\begin{theorem}\label{thm:main}
Given an $n$-vertex graph $G$ and $t \geq 1$, one can in polynomial time either:
\begin{itemize}
\item output an induced copy of $S_{t,t,t}$ in $G$, or
\item output a set $\mathcal{P}$ consisting of at most $11 \log n+6$ induced paths in $G$,
  each of length at most $t+1$, and a rigid extended strip decomposition of
  $G - N[\bigcup_{P \in \mathcal{P}} V(P)]$ whose every particle has at most $n/2$ vertices. 
\end{itemize}
\end{theorem}

Combining \cref{thm:main} with previously known algorithmic techniques,
we derive two algorithms for MWIS in $S_{t,t,t}$-free graphs.
Actually, our algorithms work in a slightly more general setting.
For integers $s,t \geq 1$, by $sS_{t,t,t}$ we denote the graph with $s$ connected components, each isomorphic to $S_{t,t,t}$.
Recall that by the observation of Alekseev~\cite{alekseev1982effect,Alekseev03} the only graphs $H$, for which we can hope for tractability results for MWIS in $H$-free graphs, are forests whose every component has at most three leaves.
We observe that each such $H$ is contained in $sS_{t,t,t}$, for some $s$ and $t$ depending on $H$.
Thus algorithms for $sS_{t,t,t}$-free graphs, for every $s$ and $t$, cover \emph{all} potential positive cases.

First, we observe that the statement of \cref{thm:main}
seamlessly combines with the method how~\cite{DBLP:journals/algorithmica/BacsoLMPTL19}
obtained a subexponential-time algorithm for \textsc{MWIS} in $P_t$-free graphs.
As a result, we obtain a subexponential-time algorithm for \textsc{MWIS} in $sS_{t,t,t}$-free graphs
with improved running time as compared to~\cite{DBLP:journals/corr/abs-1907-04585, DBLP:conf/soda/ChudnovskyPPT20}.

\begin{theorem}\label{thm:subexp}
Let $s,t \geq 1$ be constants.
Given an $n$-vertex $sS_{t,t,t}$-free graph $G$ with weights on vertices, one can 
in time exponential in $\Oh(\sqrt{n} \log n)$ compute an independent set in $G$ of maximum possible weight.
\end{theorem}

Second, we observe that the statement of \cref{thm:main}
again seamlessly combines with the method how~\cite{DBLP:journals/corr/abs-1907-04585, DBLP:conf/soda/ChudnovskyPPT20} obtained a QPTAS for \textsc{MWIS} in $P_t$-free graphs, 
obtaining an arguably simpler QPTAS for \textsc{MWIS} in $sS_{t,t,t}$-free graphs
with improved running time 
(compared to~\cite{DBLP:journals/corr/abs-1907-04585, DBLP:conf/soda/ChudnovskyPPT20}).
\begin{theorem}\label{thm:qptas}
Let $s,t \geq 1$ be constants.
Given an $n$-vertex $sS_{t,t,t}$-free graph $G$ with weights on vertices, and a real $\varepsilon > 0$, one can 
in time exponential in $\Oh(\varepsilon^{\exponentqptaseps} \log^{\exponentqptaslogn} n)$ compute
an~independent set in $G$ that is within a~factor of $(1-\varepsilon)$ of the maximum possible weight.
\end{theorem}

After preliminaries in \cref{sec:prelim},
we prove \cref{thm:main} in \cref{sec:main}.
Proofs of Theorems~\ref{thm:subexp} and~\ref{thm:qptas} are provided in \cref{sec:algo}.
Finally, we discuss future steps in \cref{sec:conclusion}.

%% file: prelim.tex
\subparagraph*{Notation.} 
For a family $\mathcal{Q}$ of sets, by $\bigcup \mathcal{Q}$ we denote $\bigcup_{Q \in \mathcal{Q}} Q$.
If the base of a logarithmic function is not specified, we mean the logarithm of base 2, i.e., $\log n := 
\log_2 n$.
For a function $\wei : V \to \Z$ and subset $V' \subseteq V$, we denote $\wei(V'):=\sum_{v \in V} \wei(v)$.

Let $G$ be a graph. For $X \subseteq V(G)$, by $G[X]$ we denote the subgraph of $G$ induced by $X$, i.e., $(X, \{uv \in E(G) : u,v \in X\})$.
If the graph $G$ is clear from the context, we will often identify induced subgraphs with their vertex sets.
The sets $X,Y \subseteq V(G)$ are \emph{complete} to each other if for every $x \in X$ and $y \in Y$ the edge $xy$ is present in $G$.
Note that this, in particular, implies that $X$ and $Y$ are disjoint.
We say that two sets $X,Y$ \emph{touch} if $X \cap Y \neq \emptyset$ or there is an edge with one end in $X$ and another in $Y$.

For a vertex $v$, by $N_G(v)$ we denote the set of neighbors of $v$, and by $N_G[v]$ we denote the set $N_G(v) \cup \{v\}$.
For a set $X \subseteq V(G)$, we also define $N_G(X) := \bigcup_{v \in X} N(v) \setminus X$, and $N_G[X] = N_G(X) \cup X$.
If it does not lead to confusion, we omit the subscript and write simply $N( \cdot)$ and $N[ \cdot]$.

By $T(G)$, we denote the set of all triangles in $G$.
Similarly to writing $xy \in E(G)$, we will write $xyz \in T(G)$ to indicate that $G[\{x,y,z\}] \simeq K_3$.

\subparagraph*{Extended strip decompositions.} 
Now let us define a certain graph decomposition which will play an important role in the paper.
An \emph{extended strip decomposition} of a graph $G$ is a pair $(H, \eta)$ that consists of:
\begin{itemize}
\item a simple graph $H$,
\item a set $\eta(x) \subseteq V(G)$ for every $x \in V(H)$,
\item a set $\eta(xy) \subseteq V(G)$ for every $xy \in E(H)$, and its subsets $\eta(xy,x),\eta(xy,y) \subseteq \eta(xy)$,
\item a set $\eta(xyz) \subseteq V(G)$  for every $xyz \in T(H)$,
\end{itemize}
which satisfy the following properties (see also see \cref{fig:esd}):
\begin{enumerate}
\item $\{\eta(o)~|~o \in V(H)\cup E(H) \cup T(H)\}$ is a partition of $V(G)$,
\item for every $x \in V(H)$ and every distinct $y,z \in N_H(x)$, the set $\eta(xy,x)$ is complete to $\eta(xz,x)$,
\item every $uv \in E(G)$ is contained in one of the sets $\eta(o)$ for $o \in V(H) \cup E(H)\cup T(H)$, or is as follows:
\begin{itemize}
\item $u \in \eta(xy,x), v\in \eta(xz,x)$ for some $x \in V(H)$ and $y,z \in N_H(x)$, or
\item $u \in \eta(xy,x), v\in \eta(x)$ for some $xy \in E(H)$, or
\item $u \in \eta(xyz)$ and $v\in \eta(xy,x) \cap \eta(xy,y)$ for some $xyz \in T(H)$. 
\end{itemize}
\end{enumerate}

\begin{figure}[htb]
\includegraphics[scale=1,page=1]{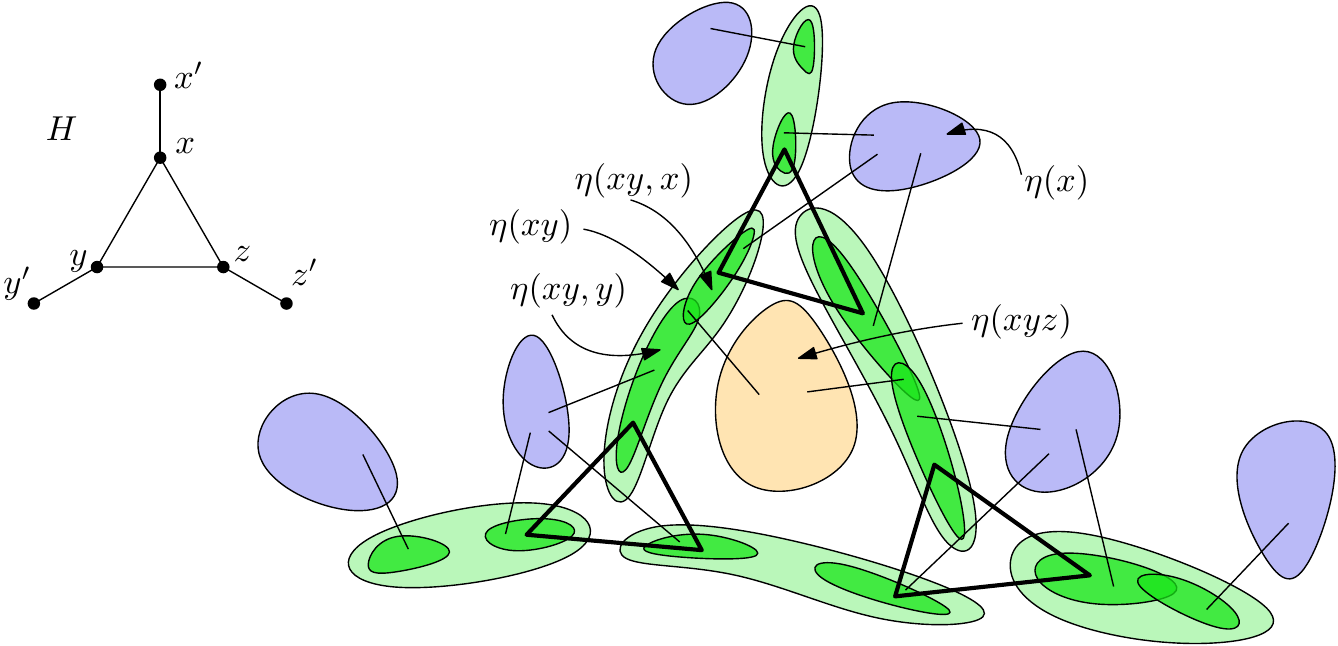}
\caption{A graph $H$ and as extended strip decomposition $(H,\eta)$ of a graph $G$.
Sets $\eta(\cdot)$ corresponding to vertices, edges, and the triangle of $H$ are marked green, blue, and orange, respectively.
The edges between distinct sets are drawn thick if they must exist, and thin if they may exist.
}\label{fig:esd}
\end{figure}

Note that for an extended strip decomposition $(H,\eta)$ of a graph $G$, the number of vertices of $H$ can be much larger than the number of vertices of $G$.
However, in such case many sets $\eta(\cdot)$ are empty and thus $H$ is ``unnecessarily complicated.''
An extended strip decomposition $(H,\eta)$ is \emph{rigid} if (i) for every $xy \in E(H)$ it holds that $\eta(xy,x) \neq \emptyset$,
and (ii) for every $x \in V(H)$ such that $x$ is an isolated vertex it holds that $\eta(x) \neq \emptyset$.\sv{\footnote{Proofs of statements marked by $\spadesuit$ are postponed to the appendix.}}

\lv{\begin{restatable}{observation}{hBounds}\label{obs:h-bounds}}
\sv{\begin{restatable}[$\spadesuit$]{observation}{hBounds}\label{obs:h-bounds}}
Let $(H,\eta)$ be a rigid extended strip decomposition of an $n$-vertex graph $G$.
Then $|E(H)| \leq n$ and $|V(H)| \leq 2n$.
\end{restatable}

\toappendix{%
  \sv{\hBounds*}
\begin{proof}
Recall that since $(H,\eta)$ is rigid, for every $xy \in E(H)$ we have that $\emptyset \neq \eta(xy,x) \subseteq \eta(xy)$,
and for every isolated vertex $x$ of $H$ we have $\eta(x) \neq \emptyset$.

Let $V_0$ and $V_+$ denote, respectively, the sets of vertices of $H$ with degree 0 and more than 0.
As the family $\{ \eta(xy) ~|~ xy \in E(H) \} \cup \{ \eta(x) ~|~ x \in V_0 \}$ consists of pairwise disjoint nonempty subsets of $V(G)$, we conclude that $|E(H)| + |V_0| \leq n$ and therefore $|E(H)| \leq n$.

Note that by the handshaking lemma we have $|E(H)| \geq |V_+|/2$, and so 
$|V(H)| = |V_0| + |V_+| \leq |V_0| + 2|E(H)| \leq 2n$ by the previous argument.
\end{proof}
}

We say that a vertex $v \in V(G)$ is \emph{peripheral} in $(H,\eta)$ if there is a degree-one vertex $x$ of $H$,
such that $\eta(xy,x)=\{v\}$, where $y$ is the (unique) neighbor of $x$ in $H$.
For a set $Z \subseteq V(G)$, we say that $(H, \eta)$ is an \emph{extended strip decomposition of $(G,Z)$} if $H$ has $|Z|$ degree-one vertices and each vertex of $Z$ is peripheral in $(H,\eta)$.

The following theorem by Chudnovsky and Seymour~\cite{DBLP:journals/combinatorica/ChudnovskyS10}
is a slight strengthening of their celebrated solution of the famous \emph{three-in-a-tree} problem.
We will use it as a black-box to build extended strip decompositions.

\begin{theorem}[{Chudnovsky, Seymour~\cite[Section 6]{DBLP:journals/combinatorica/ChudnovskyS10}}]\label{thm:three-in-a-tree}
Let $G$ be an $n$-vertex  graph and consider $Z \subseteq V(G)$ with $|Z| \geq 2$.
There is an algorithm that runs in time $\Oh(n^5)$ and returns one of the following:
\begin{itemize}
\item an induced subtree of $G$ containing at least three elements of $Z$,
\item a rigid extended strip decomposition $(H,\eta)$ of $(G,Z)$.
\end{itemize}
\end{theorem}

Let us point out that actually, an extended strip decomposition produced by \cref{thm:three-in-a-tree} satisfies more structural properties,
but of our purpose, we will only use the fact that it is rigid.

\subparagraph*{Particles of extended strip decompositions.}
Let $(H, \eta)$ be an extended strip decomposition of a graph $G$.
We introduce some special subsets of  $V(G)$ called \emph{particles}, divided into five \emph{types}.
\begin{align*}
\textrm{vertex particle:} &\quad A_{x} := \eta(x) \text{ for each } x \in V(H)\\
\textrm{edge interior particle:} &\quad A_{xy}^{\perp} := \eta(xy) \setminus (\eta(xy,x) \cup \eta(xy,y)) \text{ for each } xy \in E(H),\\
\textrm{half-edge particle:} &\quad A_{xy}^{x} :=  \eta(x) \cup \eta(xy) \setminus \eta(xy,y) \text{ for each } xy \in E(H),\\
\textrm{full edge particle:} &\quad A_{xy}^{xy} := \eta(x) \cup \eta(y) \cup \eta(xy) \cup \bigcup_{z ~:~ xyz \in T(H)} \eta(xyz) \text{ for each } xy \in E(H),\\
\textrm{triangle particle:} &\quad A_{xyz} := \eta(xyz) \text{ for each } xyz \in T(H).
\end{align*}

Observe that the number of all particles of $(H,\eta)$ is at most $\Oh(|V(H)|^3)$.
However, the number of nonempty particles is linear in the number of vertices of $G$.

\sv{\begin{restatable}[$\spadesuit$]{observation}{particleBound}\label{obs:particle-bound}}
\lv{\begin{restatable}{observation}{particleBound}\label{obs:particle-bound}}
Let $(H,\eta)$ be an extended strip decomposition of an $n$-vertex graph.
Then the number of nonempty particles of $(H, \eta)$ is bounded by $4n$.
\end{restatable}
\toappendix{%
  \sv{\particleBound*}
\begin{proof}
Let $V',E',T'$, respectively, the subsets consisting of those elements $o$ of $V(H)$, $E(H)$, or $T(H)$,
for which $\eta(o) \neq \emptyset$.
Observe that each $o \in V' \cup T'$ gives rise to one nonempty particle $A_o$,
and each $xy \in E'$ gives rise to at most four nonempty particles: $A_{xy}^{\perp}, A_{xy}^x,A_{xy}^y,A_{xy}^{xy}$.
Moreover, since $\{\eta(o)~|~o \in V' \cup E' \cup T'\}$ are pairwise disjoint subsets of $V(G)$,
we have that $|V'|+|E'|+|T'| \leq n$.
Hence, the number of nonempty particles is bounded by $|V'|+|T'|+4|E'| =(|V'|+|T'|+|E'|) + 3|E'| \leq 4n.$
\end{proof}
}

A vertex particle $A_x$ is \emph{trivial} if $x$ is an isolated vertex in $H$.
Similarly, an extended strip decomposition $(H,\eta)$ is \emph{trivial} if $H$ is an edgeless graph. The following observation follows immediately from the definitions of an extended strip decomposition and particles.

\begin{observation}\label{obs:particlenbrhood}
Let $(H,\eta)$ be an extended strip decomposition of a graph $G$.
For each $xy \in E(H)$ the following hold:
\begin{enumerate}
\item $A^{\perp}_{xy} \subseteq A_{xy}^x \subseteq A_{xy}^{xy}$,
\item for any $v_x \in \eta(xy,x)$ and $v_y \in \eta(xy,y)$ we have $N(A_{xy}^{xy}) = N(v_x) \cup N(v_y) \setminus A_{xy}^{xy}$.
\end{enumerate}
\end{observation}

We conclude this section by recalling an important property of particles of extended strip decompositions,
observed by Chudnovsky et al.~\cite{DBLP:journals/corr/abs-1907-04585}.

\begin{theorem}[{Chudnovsky et al.~\cite[Lemma 6.8]{DBLP:journals/corr/abs-1907-04585}}]\label{thm:meet}
Let $(H,\eta)$ be an extended strip decomposition of $G$.
Suppose $P_1,P_2,P_3$ are three induced paths in $G$ that do not touch each other, and moreover each of $P_1,P_2,P_3$ has an endvertex that is peripheral in $(H,\eta)$.
Then in $(H,\eta)$ there is no particle that touches each of $P_1,P_2,P_3$.
\end{theorem}

%% file: main-res.tex
\newcommand{\tinythr}{\ensuremath{t}}
\newcommand{\smallthr}{\ensuremath{t+1}}
\newcommand{\thr}{\ensuremath{3t+1}}
\newcommand{\decr}{\ensuremath{\frac{2}{3}}}
\newcommand{\mult}{11}
\newcommand{\recurse}[1]{\ensuremath{6\log_{3/2}\left(\left|#1\right|\right)}+6}
\newcommand{\maxlen}{\ensuremath{maxlen}}
\newcommand{\lin}{\ensuremath{lin}}
\newcommand{\numb}{6}
\newcommand{\Qq}{\mathcal{Q}}
\newcommand{\Pp}{\mathcal{P}}
\newcommand{\Ss}{\mathcal{S}}
\newcommand{\Sst}{\mathcal{S}_{\ge t}}
\newcommand{\longg}{\mathrm{long}}
\newcommand{\shell}{\mathrm{shell}}
\newcommand{\trim}{\mathrm{pref}}
\newcommand{\trimt}{\mathrm{pref}_{\ge t}}
\newcommand{\elems}{\ensuremath{\bigcup}}
\newcommand{\esd}{extended strip decomposition\xspace}
\newcommand{\wh}[1]{\ensuremath{\widehat{#1}}}

In this section, we prove our main result, i.e., \cref{thm:main}.
Let us first give an~overview of our approach. We present a recursive algorithm that, for a given graph $G$, will return one of the outcomes of \cref{thm:main}.
Let $n\coloneqq|V(G)|$ be the number of vertices in the input graph; the value of $n$ will not change throughout the recursive steps of the algorithm.
We start with finding a Gyárfás path $Q$ navigating towards the largest component in $G$.
That is, by \cref{thm:gyarfas}, we find $Q$ such that each connected component of $G-N[Q]$ is of size at most $\frac{n}{2}$.
Finding such a small connected component is a great outcome as we can readily include it as a small trivial vertex particle of an \esd we are constructing.
We say that a particle is \emph{small} if its size is at most $\frac{n}{2}$, and an \esd is \emph{refined} if all its particles are small.
Observe that if $|Q|\le\thr$, we immediately get the desired refined \esd of $G$.
Otherwise, we proceed to the main part of the algorithm. At each step, we will remove some vertices from $Q$, and will measure the progress of our algorithm in the number of the remaining vertices of $Q$.

Formally, we create a set $\Qq$ of pairwise not touching induced paths such that $\elems{\Qq}\subseteq Q$  and $|\Qq|\le 2$.
At each step of recursion we obtain a set $\hat{\Qq}$ with $|\elems{\hat{\Qq}}|\le \decr |\elems{\Qq}|$ that represents $\Qq$ for the next step.
Hence, in $\mult \log n$ recursive steps, $|\elems{\Qq}|$ drops below \thr.
In the base case of the recursion, when $|\elems{\Qq}|\le\thr$, we return the refined trivial \esd ensured by maintaining the property that $G-N[\elems{\Qq}]$ has connected components of size at most $\frac{n}{2}$ throughout the recursive steps.
In each step of recursion, we further split the induced path(s) in $\Qq$ so we are able to use \cref{thm:three-in-a-tree} to obtain an \esd $(H,\eta)$.
If $(H,\eta)$ is already refined, then we are done. Otherwise, it contains a particle $A$ that is not small. We use \cref{thm:meet} to select at most two paths touching $A$. Then it is easy to separate $A$ with the respective touching paths from the rest of the graph.
The graph induced by $A$ and the touching paths form a smaller instance, i.e., an instance where $|\elems{\Qq}|$ drops by a~factor of $\decr$.
We need to ensure that at every recursive step, we include only a constant number of paths of length $t+1$ into $\Pp$ (i.e., the set of paths in the second outcome of \cref{thm:main}). %
We now prove the core recursive formulation of the algorithm formally.

\begin{lemma}[Recursion]\label{lem:recursion}
  Given a graph $G$ and a set $\Qq$ of at most two induced paths (vertex disjoint non-adjacent), and a refined \esd of $G-N[\elems{\Qq}]$. 
  In polynomial time, we can output one of the following:
  \begin{itemize}
    \item an induced copy of $S_{t, t, t}$ in $G$, or
    \item $\Pp$, $X\subseteq N[\elems\Pp]$, and a refined \esd $(G-X,\eta)$, so that $|\Pp|\le \recurse{\elems{\Qq}}$ and the longest path in $\Pp$ has at most $\smallthr$ vertices.
  \end{itemize}
\end{lemma}

\begin{proof}
  If the longest path of $\Qq$ has at most $\thr$ vertices, return $\Pp\coloneqq\Qq$ where each path in $\Pp$ may be further split in at most three paths on at most $\smallthr$ vertices, and $X\coloneqq N[\elems\Pp]$. 
Hence, we output the \esd we were given by the assumptions of the lemma.

Otherwise, let $Q_1$ be the longest path in $\Qq$.
Let $u_1$ and $u_2$ be the $\left(\left\lfloor \frac{|Q_1|}{3} \right\rfloor + 1\right)$-th and the $\left(2\left\lfloor \frac{|Q_1|}{3} \right\rfloor + 2\right)$-th vertex of $Q_1$, respectively.
The removal of $u_1$ and $u_2$ from $Q_1$ divides the path into three induced non-touching subpaths $Q_1^1$, $Q_1^2$, and $Q_1^3$, each of length at least $t$.
Let $Q_2$ be the remaining path of $\Qq$, should it exist.
We define $\Ss \coloneqq \{Q_1^1, Q_1^2, Q_1^3, Q_2\}$ if $Q_2$ exists, or $\Ss \coloneqq \{Q_1^1, Q_1^2, Q_1^3\}$, otherwise.
Consult \cref{fig:shells} to see an overview of the definitions described in this paragraph.
For each path $P \in \Ss$ we define $\trim(P)$ as the set comprising:
\begin{itemize}
 \item first $t - 1$ vertices of $P$ (or all vertices of $P$ if $|P|<t-1$), and
 \item the separating vertex of $Q_1$ directly preceding $P$ if $P \in \{Q_1^2, Q_1^3\}$. %
\end{itemize}
It can be easily seen that the set of vertices $\trim(P)$ forms an induced path of length at most $t$.
We finally define \emph{shells} of paths in $\Ss$. 
Given a~path $P \in \Ss$, we set $\shell(P) \coloneqq N[\trim(P)] \setminus \bigcup \Ss$ if $|P|\ge t$ and $\shell(P) \coloneqq N[\trim(P)]$ otherwise.
Intuitively, if $|P|< t$, the shell of $P$ takes the whole neighborhood as we do not have a use for a short path in the next stage of our algorithm.
For a long enough path $P$, the shell of $P$ intersects all short paths connecting the first vertex of $P$ with the rest of the graph.
Thus, each path from the first vertex of $P$ to any vertex of $G-\shell(P)$ outside of $P$ will have length at least $t$.
To ease the notation, we define $\Sst\coloneqq \{P\in\Ss\mid |P|\ge t\}$, $\shell(\Ss) \coloneqq \bigcup_{P \in \Ss} \shell(P)$, and $\trim(\Ss) \coloneqq \bigcup_{P \in \Ss} \trim(P)$.

\begin{figure}[htb]
\includegraphics[scale=1,page=2]{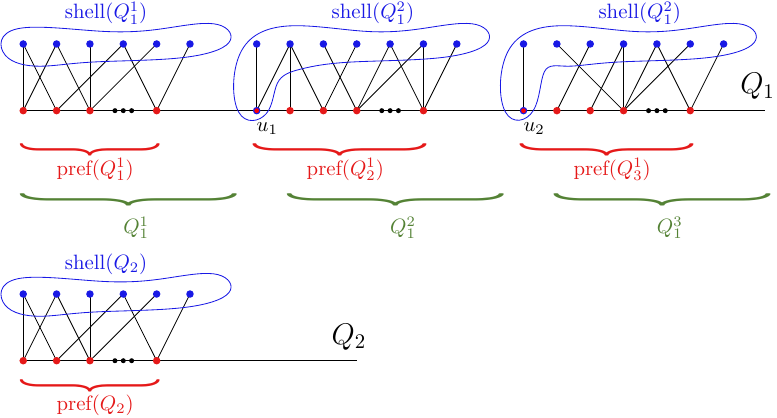}
\caption{Definitions of $\trim(\Ss)$ and $\shell(\Ss)$ in case of $|Q_2|\ge t$.} \label{fig:shells}
\end{figure}

Now, we use the algorithm from \cref{thm:three-in-a-tree} on $Z$ being the set of the first vertices of paths in $\Sst$ and the graph defined as $G-\shell(\Ss)$. 
If \cref{thm:three-in-a-tree} produced an induced tree with three leaves among $Z$, we return it as an induced $S_{t,t,t}$, since those must have been induced branches at least $t$ vertices long in $G-\shell(\Ss)$.
Hence, we obtained an \esd $(H',\eta')$ of $G-\shell(\Ss)$.
If the obtained decomposition is refined, we return $\Pp\coloneqq\trim(\Ss)$, $X\coloneqq\shell(\Ss)$, and the \esd $(H\coloneqq H',\eta\coloneqq\eta')$.

Therefore, the obtained \esd $(H',\eta')$ of $G-\shell(\Ss)$ contains a~particle $A$ which is not small, i.e., $A$ is composed of at least $\frac{n}{2}$ vertices.
As $Z$ is peripheral, we know that no three paths in $\Sst$ touch one particle by \cref{thm:meet}.
Therefore, we take the set $\hat\Qq$ of at most two paths, say $P_1$ and $P_2$, touching $A$ (for convenience, let $P_1$ or $P_2$ be an~empty set if it does not exist).
We now compute the maximum proportion of $\elems\Qq$ put to $\hat\Qq$.
If both $P_1,P_2\subseteq Q_1$, then this fraction is at most $\decr$ as by the definition $|Q_1^i|\le\frac{|Q_1|}{3}$, for $i\in\{1,2,3\}$.
If one is $Q_2$ and the other comes from $Q_1$, then we estimate $a+\frac{1-a}{3}=\frac{2a+1}{3}\le \decr$ for $a={|Q_2|}/{|\elems\Qq|}\le\frac{1}{2}$.
Hence, we know that $|\elems\hat\Qq|\le\decr |\elems\Qq|$.
We define $\hat{G}\coloneqq A\cup P_1\cup P_2$ to use \cref{lem:recursion} on a smaller instance.
Now, we need to verify that the assumption of the lemma holds.
We claim the following:

\begin{claim}
 $\hat{G}-N[\elems{\hat{\Qq}}]$ has a refined \esd.
\end{claim}

\begin{claimproof}
As $\hat{G}$ is an induced subgraph of $G$ and $G-N[\elems{\Qq}]$ has a refined \esd, we know that $\hat{G}-N[\elems{\Qq}]$ has a refined \esd.
 First, recall that $N[u_1]\setminus(Q_1^1\cup Q_1^2)\subseteq \shell(Q_1^2)$, which is disjoint with $V(\hat{G})$.
 Analogously $N[u_2]\setminus(Q_1^2\cup Q_1^3)$ is disjoint with $V(\hat{G})$.
 Also, if $|Q_2|<t$ then $Q_2$ is disjoint with $V(\hat{G})$ as well.
 Hence, $\hat{G}-N[\elems{\Qq}]\simeq \hat{G}-N[\elems{\Sst}]$.
 Also, recall that the only paths among $\Sst$ that touch $A$ are in $\hat{\Qq}$.
  Hence, observe that $\hat{G}-N[\elems{\Sst}]\simeq \hat{G}-N[\elems{\hat{\Qq}}]$. 
\end{claimproof}

Therefore, we can apply \cref{lem:recursion} inductively on $\hat{G}$ and $\hat{\Qq}$, obtaining $\hat{\Pp}$ and $\hat{X}$, and a~refined \esd~$(\hat{H},\hat{\eta})$ of $\hat{G}-\hat{X}$.
We need to combine the \esd obtained from the recursion with the \esd $(H',\eta')$ we obtained earlier.

We can always suppose that particle $A$ is of type $A^{xy}_{xy}$ for some edge $xy\in E(H')$, unless $A$ is of type $A_x$ for an isolated vertex $x\in V(H')$.
That is because  $A^{xy}_{xy}$ is the superset of all possible particle types.
As \cref{thm:three-in-a-tree} gives us that both $\eta'(xy,x)$ and $\eta'(xy,y)$ are nonempty, we can select $v_x\in\eta'(xy,x)$ and $v_y\in\eta'(xy,y)$ (possibly $v_x=v_y$).
By~\cref{obs:particlenbrhood}, the set \[X'\coloneqq\left(N(v_y)\cup N(v_x)\right)\setminus V(A)\] separates $A$ from the rest of $G$.  
Set $\Pp'\coloneqq\left\{\{v_x\},\{v_y\}\right\}$.
In the case of $A_x$ such that $x\in V(H)$ is an independent vertex, we set $\Pp'\coloneqq\emptyset$ and $X'\coloneqq\emptyset$ and still such $A$ is separated from the rest of $G$ by $X'$.  
We return:
\begin{itemize}
  \item $\Pp\coloneqq\hat{\Pp}\cup \Pp' \cup \trim(\Ss)$,
  \item $X\coloneqq\hat{X}\cup X'\cup \shell(\Ss)$,
  \item an~\esd $(H,\eta)$ of $G-X$, where $H$ is $\hat{H}$ with an additional isolated vertex $w$, and $\eta$ is $\hat{\eta}$ with an additional trivial vertex particle $\eta(w)$ containing all vertices of $G-X-A$.
\end{itemize}

We compute that $|\Pp|\le\numb + \recurse{\elems{\hat{\Qq}}} \le \recurse{\elems{\Qq}}$ as we added at most six new paths into $\Pp$.
Observe that the described algorithm runs in polynomial time as we just computed that the depth of recurrence is logarithmic in $|\elems(\Qq)|\le |V(G)|$ and each recursive call takes polynomial time in the size of $G$.

\end{proof}

\begin{proof}[Proof of \cref{thm:main}]
  Using \cref{thm:gyarfas} we find a Gyárfás path $Q$.
  We get the desired outcome by \cref{lem:recursion} on $G$ with $\Qq\coloneqq\{Q\}$.
  The \esd needed by the lemma's assumption is trivial.
  That is, each connected component of $G-Q$ is represented by a vertex particle of small size.
We conclude the proof of \cref{thm:main} by the following calculation: 
\[
  6 \log_{3/2} n + 6  \le 11 \log n+6.
\]

Note that for any \esd $(H,\eta)$ we can easily add the assumption that sets $\eta(xy,x)\neq \emptyset$ for any edge $xy\in E(H)$. 
As suppose $\eta(xy,x)=\emptyset$; then we can update $(H,\eta)$ by adding $\eta(xy)$ to $\eta(y)$ and removing $xy$ from $H$.
Moreover, we can simply remove any empty trivial vertex particle form $\eta$ and the corresponding isolated vertex from $H$.
Therefore, we may suppose that the obtained \esd is rigid.
\end{proof}

In the following simple corollary we apply \cref{thm:main} to $s S_{t,t,t}$-free graphs, for some $s,t\ge 1$.

\sv{\begin{restatable}[$\spadesuit$]{corollary}{forestFree}\label{cor:forestfree}}
\lv{\begin{restatable}{corollary}{forestFree}\label{cor:forestfree}}
Let $s\geq 1, t\geq 1$ be constants. %
Let $G$ be an $s S_{t,t,t}$-free graph on $n$ vertices.
Then in polynomial time we can find a set $X$ consisting of at most $(s-1)(3t+1) + (11 \log n + 6)(t+1)$ vertices
and a rigid extended strip decomposition of $G - N[X]$ whose every particle has at most $n/2$ vertices.
\end{restatable}
\toappendix{%
  \sv{\forestFree*}
  \begin{proof}
Induction on $s$. If $s=1$, then we obtain the result immediately by \cref{thm:main}.
Thus let us assume that $s \geq 2$ and the theorem holds for $(s-1)S_{t,t,t}$-free graphs.

We exhaustively check if there is some $Y \subseteq V(G)$ with $|Y| = 3t+1$, such that $G[Y] \simeq S_{t,t,t}$; we can do it in time $n^{3(t+1) + \Oh(1)}=n^{\Oh(1)}$.
If such $Y$ does not exist, then we can immediately apply \cref{thm:main}, and the proof is complete.
Thus suppose that $Y$ exists.

We observe that the graph $G':=G-N[Y]$ is $(s-1)S_{t,t,t}$-free. Denote $n' := |V(G')|$.
By the inductive assumption, in time $(n')^{\Oh(1)}=n^{\Oh(1)}$
we can obtain a set $X' \subseteq V(G')$ of size at most $(s-2)(3t+1) + (11 \log n' + 6)(t+1)$ and a rigid extended strip decomposition $(H,\eta)$ of $G'-N[X']$ whose every particle is of size at most $n'/2$.

We set $X = Y \cup X'$. Now $X$ and $(H,\eta)$ satisfy the statement of the theorem, as $G' - N[X'] = G - N[X]$ and $n' \leq n$. The total running time is polynomial in $n$ as the depth of the recursion is $s-1$.
\end{proof}
}

%% file: algorithmic.tex
In this section we will show how to combine \cref{thm:main} with the approach of Chudnovsky et al.~\cite{DBLP:journals/corr/abs-1907-04585, DBLP:conf/soda/ChudnovskyPPT20} in order to obtain a QPTAS and a subexponential-time algorithm for MWIS in $S_{t,t,t}$-free graphs, i.e., we prove Theorems~\ref{thm:subexp} and~\ref{thm:qptas}.

Both algorithms follow the same general outline; let us sketch it before we get into the details of each particular case.
Each algorithm is a recursive procedure, which consists of two phases.
In the first one, we deal with the vertices of $G$ that are \emph{heavy}, which means that their neighborhood is ``large'',
where the exact meaning of ``large'' depends on the particular algorithm.

Once there are no heavy vertices, i.e., the neighborhood of each vertex is ``small'', we proceed to the second phase.
We call \cref{cor:forestfree} for the current instance $G$, obtaining a~small-sized set $X$
and a rigid extended strip decomposition $(H,\eta)$ of $G-N[X]$, whose every particle is of small size.
The crux is that since we are in the second phase, all vertices in $X$ are not heavy,
and since $X$ is of small size, the whole set $N[X]$ is ``small''.
We treat $N[X]$ separately in a way that depends on the particular algorithm.

Next, for each particle $A$ of $(H,\eta)$, we call the algorithm recursively for $G[A]$,
obtaining (a good approximation of) a maximum-weight independent set in $G[A]$.
Finally, we combine the obtained results to derive (a good approximation of) a maximum-weight independent set in $G$.
This last step is based on the idea of Chudnovsky et al.~\cite{DBLP:journals/corr/abs-1907-04585, DBLP:conf/soda/ChudnovskyPPT20} to reduce the problem to finding
a maximum-weight matching in a graph obtained by a simple modification of $H$.
Since the size of $H$ is linear in $|V(G)|$ (by~\cref{obs:h-bounds}), this problem can be solved in time polynomial in $|V(G)|$
using, e.g., the classic algorithm of Edmonds~\cite{Edmonds}.
The last step is encapsulated in the following lemma, whose exact statement comes from Abrishami et al.~\cite{AbrishamiCPRS21-arxiv}.%

\begin{lemma}[Chudnovsky et al.~\cite{DBLP:journals/corr/abs-1907-04585, DBLP:conf/soda/ChudnovskyPPT20}]
\label{lem:reduction-matching}
Let $\varsigma \in [0,1]$ be a real number.
Let $G$ be an $n$-vertex graph equipped with a weight function $\wei: V(G) \to \Z$. Suppose that $G$ is given along with an extended strip decomposition $(H, \eta)$, where $H$ has $N$ vertices. 

\noindent Let $I_0 \subseteq V(G)$ be a fixed independent set in $G$.
Furthermore, assume that for each particle $A$ of $(H, \eta)$ we are given an independent set $I(A)$ in $G[A]$ such that $\wei(I(A))\geq\varsigma \cdot \wei(I_0 \cap A)$.  
Then in time polynomial in $n+N$ we can compute an independent set $I$ in $G$ such that $\wei(I) \geq \varsigma \cdot \wei(I_0)$.
\end{lemma}

Let us stress out that the algorithm from \cref{lem:reduction-matching} does not need to know the value of $\varsigma$ or the independent set $I_0$.

The main difference between our approach and the one of Chudnovsky et al.~\cite{DBLP:conf/soda/ChudnovskyPPT20} is that we use \cref{thm:main} and its consequence, i.e., \cref{cor:forestfree}. The previous algorithms used a~similar statement but with a worse (and much more involved) guarantee on the size of $X$ and each particle. Furthermore, the way we obtain our set $X$ is significantly simpler.

\subsection{Proof of \cref{thm:subexp}}
Before we proceed to the proof, let us first explain the meaning of ``small'', and how to deal with $N[X]$ in this particular case.
Here the neighborhood of a vertex is ``small'' if it has few vertices (more specifically, at most $\sqrt{n/t}$).
In the first phase, we deal with heavy vertices $v$ (i.e., of large degree) with simple branching: we guess whether $v$ is included in our optimum solution or not. Since the degree of $v$ is large, in the first branch, we obtain significant progress, which is enough to obtain a subexponential running time.

In the second phase, since $N[X]$ is the neighborhood of $\Oh(\log n)$ vertices, each of degree $\Oh(\sqrt{n})$,
the total size of $N[X]$ is $\Oh(\sqrt{n} \log n)$.
Thus we can afford to exhaustively guess the intersection of our optimum solution with $N[X]$.

\begin{proof}[Proof of \cref{thm:subexp}]
Let $s,t \geq 1$ be constants and let $(G,\wei)$ be an instance of MWIS, where $G$ is $sS_{t,t,t}$-free and has $n$ vertices.
We observe that if $n$ is small, i.e., bounded by a constant. Then we can solve the problem by brute force.
Thus we assume that $n \geq n_0$, where $n_0$ is a constant (depending on $s$ and $t$) whose exact value follows from the reasoning below.

First, consider the case that there exists $v \in V(G)$ such that $\deg v  \geq \sqrt{n/t}$.
We branch on including $v$ in the final solution: we either delete $v$ from $G$, or we delete $N[v]$ and add $v$ to the solution returned by the recursive call.
Then we output the one of these two solutions that has a larger weight.
The correctness of this step of the algorithm is straightforward.

Hence, we can assume that for every $v \in V(G)$ it holds that $\deg v \leq \sqrt{n/t}$. 
By \cref{cor:forestfree}, since $G$ is $sS_{t,t,t}$-free, we obtain a set $X$ of size $(s-1)(3t+1) + (11\log{n}+6)(t+1) \leq 12(t+1)\log n$ (here we use that $n$ is large),
and a rigid extended strip decomposition $(H, \eta)$ of $G' = G - N[X]$ whose every particle has at most $n/2$ vertices.

We exhaustively guess an independent set $J \subseteq N[X]$; think of it as an intersection of the intended optimum solution with $N[X]$.
Consider the graph $G'' := G' - N[J]$.
We modify $(H,\eta)$ by removing the vertices from $N[J]$ from the sets $\eta(\cdot)$.
Let us call the obtained strip decomposition $(H,\eta')$; note that it might not be rigid.
We call the algorithm recursively for the subgraph $G''[A]$ for every nonempty particle $A$ of $(H,\eta')$.
Let $I(A)$ be the solution. If $A = \emptyset$, then $I(A) = \emptyset$.
By the inductive assumption $I(A)$ is a maximum-weight independent set in $G''[A]$.
Then we use \cref{lem:reduction-matching} for $\varsigma=1$ to combine the solutions into a maximum-weight independent set $I_J$ of $G''$.
Finally, we return the independent set $J \cup I_J$ whose weight is maximum over all choices of $J$.
Note that the correctness of this step is guaranteed by the exhaustive guessing of $J$ and \cref{lem:reduction-matching}.

\subparagraph*{Running time.} Let $F(n)$ denote the running time of our algorithm for $n$-vertex instances.
We prove that $F(n) = 2^{\Oh\left(\sqrt{tn}\log{n}\right)}$.
If $n < n_0$, then the claim clearly holds. So let us assume that $n \geq n_0$.

In the first case we call the algorithm for two instances, one of size $n-1$ and one of size at most $n-\sqrt{n/t}$.
Hence,
\[F(n)\leq F(n-1)+F(n-\sqrt{n/t}) = 2^{\Oh\left(\frac{n \log{n}}{\sqrt{n/t} }\right)} \leq 2^{\Oh\left(\sqrt{tn}\log{n}\right)}.\]
Here we skip the description how this recursion is solved, as it it pretty standard. For a~formal proof we refer the reader to Bacs\'o et al.~\cite[Lemma~1]{DBLP:journals/algorithmica/BacsoLMPTL19}.

It remains to analyze the running time of the step in which the maximum degree of vertices in $G$ is bounded by $\sqrt{n/t}$. 
\cref{cor:forestfree} asserts that we obtain $X$ and the rigid extended strip decomposition $(H, \eta)$ of $G'=G \setminus N[X]$ in time polynomial in $n$.
There are $2^{\Oh(\sqrt{n/t} \cdot t\log{n})}=2^{\Oh(\sqrt{nt}\log{n})}$ ways of choosing the set $J$.
In polynomial time we modify $(H,\eta)$ into $(H,\eta')$.

Observe that while $(H,\eta')$ might not be rigid, it was obtained from a rigid extended strip decomposition $(H,\eta)$ by deleting some vertices from the sets $\eta(\cdot)$. In particular, both decompositions have the same sets of particles, 
and every nonempty particle of $(H,\eta')$ is also a nonempty particle of $(H,\eta)$.
Thus by \cref{obs:particle-bound} we call the algorithm recursively for at most $4n$ nonempty particles,
each of size at most $n/2$. 
By \cref{obs:h-bounds}, the total number of particles of $(H,\eta')$ is polynomial in $n$.
Finally, having computed a maximum independent set contained in each particle, by \cref{lem:reduction-matching}, we can compute the final solution in time polynomial in $n$.
Hence, there are constants $c,c_1,c_2$, where $c \gg c_1,c_2$, such that total running time of this step is bounded by:
\begin{align}\label{eq:running-bnd}
F(n) \leq 2^{c_1 \cdot \sqrt{nt} \log{n}} \left(n^{c_2} + 4n \cdot 2^{c \cdot \sqrt{tn/2} \log{(n/2)}}\right) \overset{c \gg c_1,c_2} \leq 2^{c \cdot \sqrt{tn}\log{n}},
\end{align}
and so is the total complexity of the algorithm.
\end{proof}

\subsection{Proof of \cref{thm:qptas}}
Again let us start with explaining the algorithm-specific details of the outline presented at the start of Section~\ref{sec:algo}.

We will use the notion of $\beta$-heavy vertices from~\cite{DBLP:journals/corr/abs-1907-04585, DBLP:conf/soda/ChudnovskyPPT20}.
Consider a graph $G$, a weight function $\wei: V(G) \to \Z$, and an independent set $I \subseteq V(G)$.
Let $\beta \in (0, 1/2]$ be a real. 
We say that a vertex $v \in V(G)$ is \emph{$\beta$-heavy} (with respect to $I$) if $\wei(N [v] \cap I) > \beta \cdot \wei(I)$.
A set $J$ is \emph{good for $I$} if $J \subseteq I$ and $N[J]$ contains all vertices that are $\beta$-heavy with respect to $I$.

\begin{lemma}[Chudnovsky et al.~\cite{DBLP:journals/corr/abs-1907-04585, DBLP:conf/soda/ChudnovskyPPT20}]\label{lem:heavy-vertices}
Let $G$ be an $n$-vertex graph for $n>2$, $\wei:V(G) \to \Z$ be a weight function,
$I \subseteq V(G)$ be an independent set, and $\beta \in (0, 1/2]$ be a real.
Then there exists a set $J$ of size at most $\lceil \beta^{-1} \log{n}\rceil$ which is good for $I$.
\end{lemma}

Now the vertex is heavy if it is $\beta$-heavy for some carefully chosen parameter $\beta$.
This means that a neighborhood of a vertex is ``large'' if it contains a significant ($\geq \beta$) fraction of the weight of $I_{\mathsf{OPT}}$.
In the first phase, we exhaustively guess the set $J$ that is good for a fixed optimum solution $I_{\mathsf{OPT}}$.
Note that $J$ is of small size and since $J \subseteq I_{\mathsf{OPT}}$,
we know that $N(J)$ contains no vertices from $I_{\mathsf{OPT}}$ and thus can be safely removed from the graph.

Since $J$ is good for $I_{\mathsf{OPT}}$, we know that $G-N[J]$ contains no heavy vertices, and for this graph we call \cref{cor:forestfree}. Now, as $N[X]$ is a neighborhood of few non-heavy vertices,
we know that the total weight of $I_{\mathsf{OPT}} \cap N[X]$ is small and thus can be sacrificed, as we aim for an approximation.

\begin{proof}[Proof of \cref{thm:qptas}]
Let $s,t \geq 1$ be constants and let $(G,\wei)$ be an instance of MWIS, where $G$ is $sS_{t,t,t}$-free and has $n$ vertices. Let $\epsilon \in (0,1)$ be fixed.
Fix a maximum-weight independent set $I_{\mathsf{OPT}}$ in $G$ with respect to $\wei$.
We describe a procedure that finds in $G$ an~independent set $I$ of weight at least $(1 - \eps) \cdot \wei(I_{\mathsf{OPT}})$.

Let $N$ be the minimum power of two greater than or equal to the size of our initial instance. Note that $n \leq N < 2n$.
The value of $N$ will not change throughout the execution of the algorithm.

The algorithm itself is a recursive procedure. The arguments of each call are a graph $G'$, which is an induced subgraph of $G$, the weight function on $V(G')$ obtained by restricting the domain of $\wei$, and an integer $h$, which can be intuitively understood as the depth of the current call in the recursion tree. Since it does not lead to confusion,
we will always denote the weight function by $\wei$.
We will keep the invariant that for each call $(G',\wei,h)$ it holds that $|V(G')| \leq N/2^h$.
The initial call, corresponding to the root of the recursion tree, is for $(G,\wei,0)$. 

Consider a call for the instance $(G',\wei,h)$.
If $|V(G')| < n_0$, where $n_0$ is a constant  (depending on $s$ and $t$) that follows from the reasoning below, then we can solve the problem by brute force. Thus let us assume that $n \geq n_0$. In particular, $N > 1$.

We set 
\begin{align}\label{eq:qptas-beta}
\beta(h,\epsilon):=\frac{\epsilon}{12(t+1) \log{(N/2^h)} \cdot \left((1-\eps)\log{N} +\epsilon(h+1)\right)}.
\end{align}
It is straightforward to verify that for $h < \log N$ we have $\beta(h,\epsilon) \in (0,1/2]$.
On the other hand, if $h \geq \log N$, then $G'$ is of constant size and thus $\beta(h,\epsilon)$ is not computed for such $h$.

Let $\mathcal{J}$ be the family of all independent sets in $G'$ of size at most $\lceil\beta(h,\epsilon)^{-1}\log{(N/2^h)}\rceil$. For each $J \in \mathcal{J}$ we proceed as follows.
If $|V(G'-N[J])| < n_0$, then we compute a maximum-weight independent set $I_J$ in $G' - N[J]$ by brute force.
Otherwise, we use \cref{cor:forestfree}, to obtain a set $X_J \subseteq V(G'-N[J])$ and a rigid extended strip decomposition $(H,\eta)$ of $G'-N[J]-N[X_J]$ such that each particle of $(H, \eta)$ is of size at most $|V(G'-N[J])|/2$.
By \cref{cor:forestfree}, we obtain
\begin{align}
\begin{split}
|X_J| \leq \ & (s-1)(3t+1)+(11 \log |V(G' - N[J])|+6)(t+1) \\
\leq \ &12(t+1)\log |V(G')| \leq 12(t+1)\log (N/2^h).
\end{split}\label{eq:sizex}
\end{align}
Let $Y_J := N(J) \cup N[X_J]$.
We modify $(H,\eta)$ into an extended strip decomposition of $G'-Y_J$ as follows.
For each $v \in J$, we add to $H$ an isolated vertex $x_v$, and set $\eta(x_v)=\{v\}$.\footnote{Another possible way of dealing with the set $J$ would be to add it directly in the computed solution. However, we decided to restore $J$ to the graph, so that these vertices are handled by \cref{lem:reduction-matching} and do not require any special treatment.}
Let us call this extended strip decomposition $(H',\eta')$.
Observe that each particle of $(H',\eta')$ is of size at most $|V(G'-N[J])|/2 \leq |V(G')|/2$.
Furthermore, since $(H,\eta)$ is rigid, so is $(H',\eta')$.

For each nonempty particle $A$ of $(H',\eta')$ we call the algorithm recursively on an instance $(G'[A], \wei, h+1)$.
Let $I(A)$ be the value returned by the algorithm. For each empty particle $A$ we set $I(A) := \emptyset$.
Finally, we apply the algorithm from \cref{lem:reduction-matching}, in order to obtain an~independent set $I_J$ of $G'-Y_J$ and thus of $G'$.
Recall that the value of $\varsigma$ is not needed to apply \cref{lem:reduction-matching}; we will define it in the next paragraph when we discuss the approximation guarantee.
As the solution, we return the set $I_J$ of maximum weight, over all choices of $J \in \mathcal{J}$.

\subparagraph*{Approximation guarantee.}
Consider the recursion tree of our algorithm.
We mark some nodes of the recursion tree. First, we mark the root.
Now consider some marked node $z$ corresponding to a call $(G',\wei,h)$, such that $z$ is not a leaf node.
Observe that by \cref{lem:heavy-vertices}, there is some $J \in \mathcal{J}$ (for this particular instance)
which is good for $I_{\mathsf{OPT}} \cap V(G')$. Fix such $J$. If there is more than one, we choose one arbitrarily.
We mark the children of $z$ that correspond to the calls on the particles of the extended strip decomposition of $G' - Y_J$. 

Let $\cT$ be the subtree of the recursion tree induced by the marked nodes. 
Note that each leaf of $\cT$ is a leaf of the whole recursion tree, i.e., it corresponds to an instance of constant size.
Since at each level of the recursion, the size of the instance drops by at least half, we observe that each instance at level $h$ (where the root is at level 0) is of size at most $N/2^h$.
Consequently, the depth of $\cT$ is at most $\log N$. 

Consider a call for an instance $(G',\wei,h)$ and let $J$ be good for $I_{\mathsf{OPT}}$.  
Let us estimate $\wei(I_{\mathsf{OPT}} \cap Y_J)$.
First, observe that since $J \subseteq I_{\mathsf{OPT}}$, we have that $\wei(I_{\textsf{OPT}} \cap N(J)) = 0$.
Moreover, since $J$ was chosen to be good, there are no $\beta(h,\eps)$-heavy vertices in $V(G' - N[J])$,
and in particular, in $N[X_J]$.
Hence,
\begin{align}
\begin{split}
\wei(I_{\textsf{OPT}} \cap Y_J) = \ & \wei(I_{\textsf{OPT}} \cap N[X_J]) \leq |X_J|  \cdot \beta(h, \epsilon) \cdot \wei(I_{\textsf{OPT}} \cap V(G'))\\
\overset{\eqref{eq:qptas-beta} \text{ and } \eqref{eq:sizex}}\leq \ & \frac{\eps}{(1-\eps) \log N + \eps(h+1)}\cdot \wei(I_{\textsf{OPT}} \cap V(G')).
\end{split}\label{eq:weighty}
\end{align}

The following claim shows that the solution computed for the instance $(G',\wei,h)$ at each node of $\cT$ is a reasonable approximation of $I_{\mathsf{OPT}} \cap V(G')$.

\begin{claim}
Let $z$ be a node of $\cT$, and let $(G',\wei,h)$ be the instance corresponding to $z$.
Let $I$ be the independent set returned by the algorithm for the call at $z$.
Then $\wei(I) \geq \left(1-\epsilon +\frac{\epsilon h}{\log{N}}\right)\cdot \wei\left(I_{\mathsf{OPT}}\cap V(G')\right)$. 
\end{claim}
\begin{claimproof}
First, observe that if $z$ is a leaf of $\cT$, then the statement of the claim is satisfied.
Indeed, in this case $I$ is computed by brute force, and hence $\wei(I)=\wei(I_{\mathsf{OPT}} \cap V(G'))$.

Recall that the algorithm returns the solution of maximum weight among all choices of $J \in \mathcal{J}$,
so clearly we have $\wei(I) \geq \wei(I_J)$, where $J$ is good for $I_{\mathsf{OPT}} \cap V(G')$.

We proceed by induction on $h$. First, consider a node $z$ at the level $h=\log{N}$.
As the depth of $\cT$ is at most $\log N$, we observe that $z$ must be a leaf, so the claim follows by the observation above.

Assume that the claim holds for $h+1 \in [\log{N}]$ and consider a node $z$ at level $h$.
If $z$ is a leaf, then again, we are done.
Otherwise, let $\cA$ be the set of nonempty particles of the extended strip decomposition of $G'-Y_J$.
For every such particle $A$, we recursively computed an independent set $I(A)$.
By the inductive assumption, we have that $\wei(I(A)) \geq \left(1 - \epsilon + \frac{\epsilon(h+1)}{\log{N}}\right)\wei(I_{\textsf{OPT}} \cap V(G'[A]))$; note that these recursive calls are at level $h+1$.
Clearly, the same holds for empty particles because $\emptyset$ is there an optimum solution.

Thus, by \cref{lem:reduction-matching} applied to $I_{\mathsf{OPT}}$ and $\varsigma = 1 - \epsilon + \frac{\epsilon(h+1)}{\log{N}}$, we obtain an independent set $I_J$ in $G' - Y_J$, such that
\begin{align}
\begin{split}
\wei(I_J) \geq \ & \left(1 - \epsilon + \frac{\epsilon(h+1)}{\log{N}}\right) \wei(I_{\textsf{OPT}} \cap V(G'-Y_J))\\
= \ & \left(1 - \epsilon + \frac{\epsilon(h+1)}{\log{N}}\right) \biggl( \wei(I_{\textsf{OPT}} \cap V(G')) - \wei(I_{\textsf{OPT}} \cap Y_J)\biggr).
\end{split}\label{eq:qptas-first}
\end{align}

Combining \eqref{eq:qptas-first} with \eqref{eq:weighty} and simplifying the formula, we obtain
\begin{align*}
\wei(I_J) \geq \left(1 - \epsilon + \frac{\epsilon h}{\log{N}}\right) \wei(I_{\textsf{OPT}} \cap V(G')), 
\end{align*}
which concludes the proof of the claim.
\end{claimproof}

Since the root of the recursion tree belongs to $\cT$, the final result $I$ returned for the call at the root (i.e., for $(G,\wei,0)$) satisfies
\[
\wei(I) \geq \left(1-\epsilon\right)\cdot \wei\left(I_{\mathsf{OPT}}\cap V(G)\right) = (1-\eps) \cdot \wei(I_{\mathsf{OPT}}). 
\]
This concludes the discussion of the approximation guarantee.

\subparagraph*{Running time.}
Recall that the recursion tree has depth at most $\log N$.
Let us show the following claim concerning the running time.

\lv{\begin{restatable}{claim}{qptastime}\label{clm:qptastime}}
\sv{\begin{restatable}[$\spadesuit$]{claim}{qptastime}\label{clm:qptastime}}
Let $z$ be a node of the recursion tree, and let $(G',\wei,h)$ be the instance corresponding to $z$.
Then the algorithm solves this instance in time $2^{\Oh\left( \eps^{-1} \log^4 N \log \left(N/2^{h-1}\right)\right)}$.
\end{restatable}

\toappendix{
    \sv{\qptastime*}
    \begin{claimproof}
Let $F(h)$ denote the upper bound for the running time of our algorithm, depending on the level of the call in the recursion tree.
We aim to show that there is an absolute constant $c$, such that for $N$ sufficiently large we have 
\[F(h) \leq 2^{c \cdot \eps^{-1} \log^4 N \log \left(N/2^{h-1}\right)}.\]

Recall that $|V(G')| \leq N/2^h$.
If $z$ is a leaf, then the instance is of constant size, and thus the claim holds (assuming that $c$ is sufficiently large).
In particular this happens if $h = \log N$. So let us assume that the claim holds for the calls at level $h+1$ and that $h < \log N$.

Recall that we first enumerate the family $\mathcal{J}$ of all independent sets of size at most $\lceil \beta(h,\epsilon)^{-1} \log (N/2^h) \rceil$. Observe that
\begin{align*}
|\mathcal{J}| \leq \ & |V(G')|^{\lceil \beta(h,\epsilon)^{-1} \log (N/2^h) \rceil} \leq 2^{\log(N/2^h) \lceil \beta(h,\epsilon)^{-1} \log (N/2^h) \rceil},
\end{align*}
and the family $\mathcal{J}$ can be enumerated in time polynomial in its size.

For each $J \in \mathcal{J}$, using \cref{cor:forestfree} and modifying its outcome,
in polynomial time we obtain a set $X_J$ and a rigid extended strip decomposition $(H',\eta')$ of $G-Y_J$, where $Y_J = N[X_J] \cup N(J)$.

Next, we call the algorithm recursively for at most $4\cdot |V(G')| \leq 4 \cdot N/2^h$ instances, each at depth $h+1$.
Finally, use use \cref{lem:reduction-matching} to obtain our solution in time polynomial in $|V(G')|$ and thus in $N/2^h$.

Thus the running time is bounded by the following expression (here $c_1,c_2,c_3$ are absolute constants, such that $c_1$ and $c_2$ are much smaller than $c_3$, and $c_3 = c/12(t+1)$): 
\begin{align*}
F(h) \leq \ &  2^{c_1 \cdot \beta(h,\epsilon)^{-1} \log^2 (N/2^h) } \cdot \left ((N/2^h)^{c_2}  + 4 \cdot (N/2^h) \cdot F(h+1) \right) \\
\overset{c_3 \gg c_1,c_2 }\leq \ & 2^{ c_3 \cdot \beta(h,\epsilon)^{-1} \log^2 (N/2^h) } \cdot 2^{c \cdot \eps^{-1} \log^4 N \log (N/2^{h}) }\\
= \ & \exp \ \Bigl\{ c_3 \cdot \beta(h,\epsilon)^{-1} \log^2 (N/2^h)  +  c \cdot \eps^{-1} \log^4 N \log (N/2^{h})  \Bigr\} \\
\leq \ & \exp \ \Bigl\{ c_3 \cdot 12(t+1)  \cdot \left(\frac{1-\epsilon}{\epsilon} \log N + (h+1) \right)\log^3 (N/2^h)  + c \cdot \eps^{-1} \log^4 N \log (N/2^{h})  \Bigr\} \\
\overset{h < \log N} \leq \ & \exp \ \Bigl\{ c \cdot \epsilon^{-1} \log^4 N  + c \cdot \eps^{-1} \log^4 N \log (N/2^{h})  \Bigr\} \\
 = \ & \exp \ \Bigl\{ c \cdot \eps^{-1} \log^4 N (\log (N/2^{h})+1)  \Bigr\}
 = \exp \ \Bigl\{ c \cdot \eps^{-1} \log^4 N \log (N/2^{h-1})  \Bigr\}.
\end{align*}
This completes the proof of the claim.
\end{claimproof}
}
Now we apply \cref{clm:qptastime} to the initial call $(G,\wei,0)$ and obtain  that the overall running time is
\[
2^{\Oh(\epsilon^{-1} \log^5 N  )} = 2^{\Oh( \epsilon^{-1} \log^5 n )},
\]
as $N <2n$. This completes the proof.
\end{proof}

%% file: conc.tex
In the QPTAS of Chudnovsky, Pilipczuk, Pilipczuk, and Thomass\'{e}~\cite{DBLP:journals/corr/abs-1907-04585, DBLP:conf/soda/ChudnovskyPPT20}
it was more convenient to measure the weight of parts of the graph
not by the number of vertices, but by the weight of the intersection of the sought solution 
with the part in question. 
We observe that we can adapt Theorem~\ref{thm:main} to this setting
of unknown weight function.
\begin{theorem}\label{thm:main:weight}
Given an $n$-vertex graph $G$ and an integer $t$, one can in time $n^{\Oh(t \log n)}$ 
either:
\begin{itemize}
\item output an induced copy of $S_{t,t,t}$ in $G$, or
\item output a family $\mathcal{F}$ satisfying the following:
\begin{enumerate}
\item every element of $\mathcal{F}$ is a pair of a 
set $\mathcal{P}$ consisting of at most $11 \log n + 6$ induced paths in $G$,
  each of length at most $t+1$, and an extended strip decomposition of
  $G - N[\elems{\Pp}]$;
\item for every weight function $\wei : V(G) \to \Z$ there exists
a pair in $\mathcal{F}$ such that every particle in the extended strip decomposition
of the pair has weight at most half of the total weight of $G$;
\item the size of $\mathcal{F}$ is bounded by $n^{\Oh(\log n)}$. 
\end{enumerate}
\end{itemize}
\end{theorem}
\begin{proof}[Proof sketch.]
As observed in~\cite{DBLP:journals/corr/abs-1907-04585, DBLP:conf/soda/ChudnovskyPPT20}, in $G$
one can identify at most $n^2$ induced paths such that for every weight function $\wei : V(G) \to \Z$, at least one of the identified path is a~Gy\'{a}rf\'{a}s' path for $\wei$, 
that is, a path $Q$ such that every connected component of $G-N[Q]$ is of weight at most
half of the weight of $G$. 
Thus, we can guess the path $Q$ as in the proof in \cref{thm:main} out of at most 
$n^2$ candidates.

Then, in the recursive step in the proof of \cref{thm:main}, instead of choosing
the heavy particle to recurse on, we guess which particle is heavy (or that none exists). 
It is easy to see that any extended strip decomposition in the process will have fewer than $n$
inclusion-wise maximal particles; thus, this gives $n^{\Oh(\log n)}$ possible
outputs to enumerate.
\end{proof}

We think the $\log n$ factor in Theorem~\ref{thm:main} is an artifact of our technique,
   and is not necessary. Therefore, we pose the following conjecture.
\begin{conjecture}\label{conjecture}
For every integer $t \geq 1$ there exists a constant $\varepsilon > 0$
and an integer $s$ such that
every $S_{t,t,t}$-free graph $G$ admits a set $P \subseteq V(G)$ of size at most $s$
such that $G-N[P]$ admits a rigid extended strip decomposition whose every particle
has at most $(1-\varepsilon)|V(G)|$ vertices. 
\end{conjecture}
Abrishami, Chudnovsky, Dibek, and Rz\k{a}\.{z}ewski~\cite{ACDR21} very recently announced
a polynomial-time algorithm for \textsc{MWIS} in $S_{t,t,t}$-free graphs of bounded degree.
Their argument is quite involved and revisits the proof of the three-in-a-tree theorem~\cite{DBLP:journals/combinatorica/ChudnovskyS10}.

Confirming Conjecture~\ref{conjecture} would imply the same result almost immediately, possibly with a better running time.
Indeed, one needs to branch on $N[P]$ and recurse on the remainder of every particle of $(H,\eta)$.
The maximum degree of $H$ is bounded by a function of the maximum degree of $G$  (i.e., is a constant),
which ensures that the sum of sizes of all particles is linear in $|V(G)|$.
This in turns implies that the total complexity of the algorithm can be bounded by a polynomial function.
Note that the same approach using Theorem~\ref{thm:main} yields quasipolynomial running time bound. 

We see Theorem~\ref{thm:main} as the analog of Theorem~\ref{thm:gyarfas}
in the classes of $S_{t,t,t}$-free graphs: with its help, 
obtaining a QPTAS or a subexponential algorithm was relatively simple, 
following the ideas of~\cite{DBLP:journals/algorithmica/BacsoLMPTL19,DBLP:journals/corr/abs-1907-04585,DBLP:conf/soda/ChudnovskyPPT20}. We expect it is a first step to get a quasipolynomial-time algorithm for \textsc{MWIS}
in $S_{t,t,t}$-free graphs, similarly as Theorem~\ref{thm:gyarfas} 
is an essential ingredient of the algorithms for $P_t$-free graphs~\cite{GartlandL20,PilipczukPR21}. 
However, there is a lot of work to be done:
the way how~\cite{GartlandL20,PilipczukPR21} measure the progress of the branching
algorithm is quite intricate; furthermore, 
for the class of $C_{>t}$-free graphs (graphs excluding all cycles of length more than $t$
    as induced subgraphs, a~proper superclass of $P_t$-free graphs)
while an analog of Theorem~\ref{thm:gyarfas} is known, the corresponding 
measure of the progress of the branching algorithm is much more involved~\cite{GartlandLPPR21}.